\documentclass[a4paper,11pt]{article}

\usepackage{amsmath,amssymb,amsthm,graphicx}

\newtheorem{definition}{Definition}
\newtheorem{lemma}{Lemma}
\newtheorem{theorem}{Theorem}

\begin{document}

\title{Classically Simulating Quantum Circuits with\\ Local Depolarizing Noise}

\author{Yasuhiro Takahashi \ \ \ Yuki Takeuchi \ \ \ Seiichiro Tani\\
NTT Communication Science Laboratories, NTT Corporation\\
{\small\tt \{yasuhiro.takahashi.rb,yuki.takeuchi.yt,seiichiro.tani.cs\}@hco.ntt.co.jp}}

\date{}

\maketitle

\begin{abstract}
We study the effect of noise on the classical simulatability of quantum circuits defined by computationally 
tractable (CT) states and efficiently computable sparse (ECS) operations. Examples of such circuits, which we call 
CT-ECS circuits, are IQP, Clifford Magic, and conjugated Clifford circuits. This means 
that there exist various CT-ECS circuits such that their output probability distributions are anti-concentrated and not classically 
simulatable in the noise-free setting (under plausible assumptions). 
First, we consider a noise model where a depolarizing channel with an arbitrarily small constant rate is applied 
to each qubit at the end of computation. We show that, under this noise model, if an approximate value of the 
noise rate is known, any CT-ECS circuit with an anti-concentrated output probability distribution is 
classically simulatable. This indicates that the presence of small noise drastically affects the 
classical simulatability of CT-ECS circuits. Then, we consider an extension of the noise model where the noise rate 
can vary with each qubit, and provide a similar sufficient condition for classically simulating CT-ECS circuits 
with anti-concentrated output probability distributions.
\end{abstract}

\section{Introduction}

\subsection{Background and Main Results}

A key step toward realizing a large-scale universal quantum computer is to demonstrate 
quantum computational supremacy~\cite{Harrow}, i.e., to perform computational tasks that are classically hard. 
As such a task, many researchers have focused on simulating quantum circuits, or more concretely, 
sampling the output probability distributions of quantum circuits. They have shown that, under plausible 
complexity-theoretic assumptions, this task is classically hard for various quantum circuits that 
seem easier to implement than universal ones. However, 
these classical hardness results have been obtained in severely restricted settings, such as a 
noise-free setting with additive approximation~\cite{Bremner2,Takeuchi,Bouland,Yoganathan} and a noise setting with 
multiplicative approximation~\cite{Fujii2}: the former requires us to sample the output probability distribution of a quantum circuit 
with additive error and the latter to sample the output probability distribution of a quantum circuit under a 
noise model with multiplicative error. Thus, there is great interest in considering the above task in a more reasonable setting.

We study the classical simulatability of quantum circuits in a noise setting with additive approximation, which requires us to 
sample the output probability distribution of a quantum circuit under a noise model with additive error. This setting is 
more reasonable than the noise-free setting with additive approximation since the presence of noise is unavoidable in realistic 
situations. Moreover, our setting is more reasonable than a noise setting with multiplicative approximation in the sense that 
we adopt a more realistic notion of approximation~\cite{Aaronson1,Bremner2}, 
although noise in this paper is more restrictive than that in~\cite{Fujii2}. We consider a noise model where a depolarizing channel 
with an arbitrarily small constant rate $0 < \varepsilon < 1$, which is denoted as $D_{\varepsilon}$, 
is applied to each qubit at the end of computation. This channel leaves a qubit unaffected 
with probability $1-\varepsilon$ and replaces its state with the completely mixed one with probability $\varepsilon$. 
We call this model noise model~{\bf A}. We also consider its extension where the noise rate 
can vary with each qubit. More concretely, when a quantum circuit has $n$ qubits, $D_{\varepsilon_j}$ is 
applied to the $j$-th qubit at the end of computation for any $1 \leq j \leq n$. We call this model noise model~{\bf B}. 
These noise models are simple, but analyzing them is a meaningful step toward studying more 
general models~\cite{Gao}, such as one where noise exists before and after each gate in a quantum circuit. 
This is because, for example, this general noise model is equivalent to noise model~{\bf A} when we focus on 
instantaneous quantum polynomial-time (IQP) circuits, which are described below, with a particular type of 
intermediate noise~\cite{Bremner3}.

A representative example of a quantum circuit that is not classically simulatable (in the noise-free setting) is an 
IQP circuit, which consists of $Z$-diagonal gates sandwiched by two Hadamard 
layers. In fact, there exists an IQP circuit such that its output probability distribution is anti-concentrated 
and not classically samplable in polynomial time with certain constant accuracy in $l_1$ norm (under 
plausible assumptions)~\cite{Bremner2}. On the other hand, Bremner et al.~\cite{Bremner3} studied the classical simulatability 
of IQP circuits under noise model~{\bf A}.\footnote{Bremner et al.\ also dealt with quantum circuits for Simon's algorithm. 
Our results can be directly extended to such circuits with access to an oracle, although we omit the details for simplicity.} 
They showed that, if the {\it exact} value of the noise rate is known, any IQP circuit with an anti-concentrated 
output probability distribution is classically simulatable in the sense that the resulting probability distribution is classically 
samplable in polynomial time with arbitrary constant accuracy in $l_1$ norm. This indicates that, 
under noise model~{\bf A}, if the exact value of the noise rate is known, the presence of small 
noise drastically affects the classical simulatability of IQP circuits.

In this paper, first, under a weaker assumption on the knowledge of the noise rate, we extend Bremner et al.'s 
result to quantum circuits that are defined by two concepts: computationally 
tractable (CT) states and efficiently computable sparse (ECS) operations~\cite{van2}. Examples of such circuits, 
which we call CT-ECS circuits, are IQP circuits, Clifford Magic circuits~\cite{Yoganathan}, and 
conjugated Clifford circuits~\cite{Bouland}. This means that there exist various CT-ECS circuits such that their output 
probability distributions are anti-concentrated and not classically simulatable in the sense 
described above for IQP circuits (under plausible assumptions). Constant-depth quantum circuits~\cite{Terhal,Bremner,Juan} are 
also CT-ECS circuits and not classically simulatable, although we do not know whether their output probability 
distributions are anti-concentrated. We postpone the explanation of CT states and ECS operations 
until Section~\ref{preliminaries}, but, as depicted in Fig.~\ref{figure1}(a), 
a CT-ECS circuit on $n$ qubits is a polynomial-size quantum circuit $C=VU$ such that $U|0^n\rangle$ is CT and 
$V^\dag Z_j V$ is ECS for any $1 \leq j \leq n$, where $Z_j$ is a Pauli-$Z$ operation on the $j$-th qubit. 
After performing $C$, we perform $Z$-basis measurements on all qubits. 
The CT-ECS circuit $C$ under noise model~{\bf B} is depicted in Fig.~\ref{figure1}(b).

Our first result assumes noise model~{\bf A}, which corresponds to the case where $\varepsilon_j = \varepsilon$ 
for any $1 \leq j \leq n$ in Fig.~\ref{figure1}(b). We show that, if an {\it approximate} 
value of the noise rate is known, any CT-ECS circuit with an anti-concentrated output probability distribution is 
classically simulatable:
\begin{theorem}[informal]\label{basic}
Let $C$ be an arbitrary CT-ECS circuit on $n$ qubits such that its output probability 
distribution $p$ is anti-concentrated, i.e., $\sum_{x\in \{0,1\}^n} p(x)^2 \leq \alpha/2^n$ 
for some known constant $\alpha \geq 1$. We assume that 
a depolarizing channel with (possibly unknown) constant rate $0 < \varepsilon < 1$ is applied to each 
qubit after performing $C$, which yields the probability distribution $\widetilde{p}_{\rm A}$. Moreover, we assume that 
it is possible to choose a constant $\lambda$ such that 
$$1 \leq \frac{\varepsilon}{\lambda} \leq 1+ c,$$
where $c$ is a certain constant depending on $\alpha$. 
Then, $\widetilde{p}_{\rm A}$ is classically samplable in polynomial time with constant accuracy in $l_1$ norm.
\end{theorem}
\noindent
Throughout the paper, the base of the logarithm is 2. If $\varepsilon$ is known, we can choose $\lambda=\varepsilon$ 
in Theorem~\ref{basic}. This case with IQP circuits precisely corresponds to Bremner et al.'s 
result~\cite{Bremner3}. As described above, there exist various CT-ECS circuits such that 
their output probability distributions are anti-concentrated and not classically simulatable in the noise-free setting 
(under plausible assumptions). Thus, Theorem~\ref{basic} indicates that, under noise model~{\bf A}, if an {\it approximate} 
value of the noise rate is known, the presence of small noise drastically affects the classical simulatability of CT-ECS circuits.

\begin{figure}
\centering
\includegraphics[scale=.37]{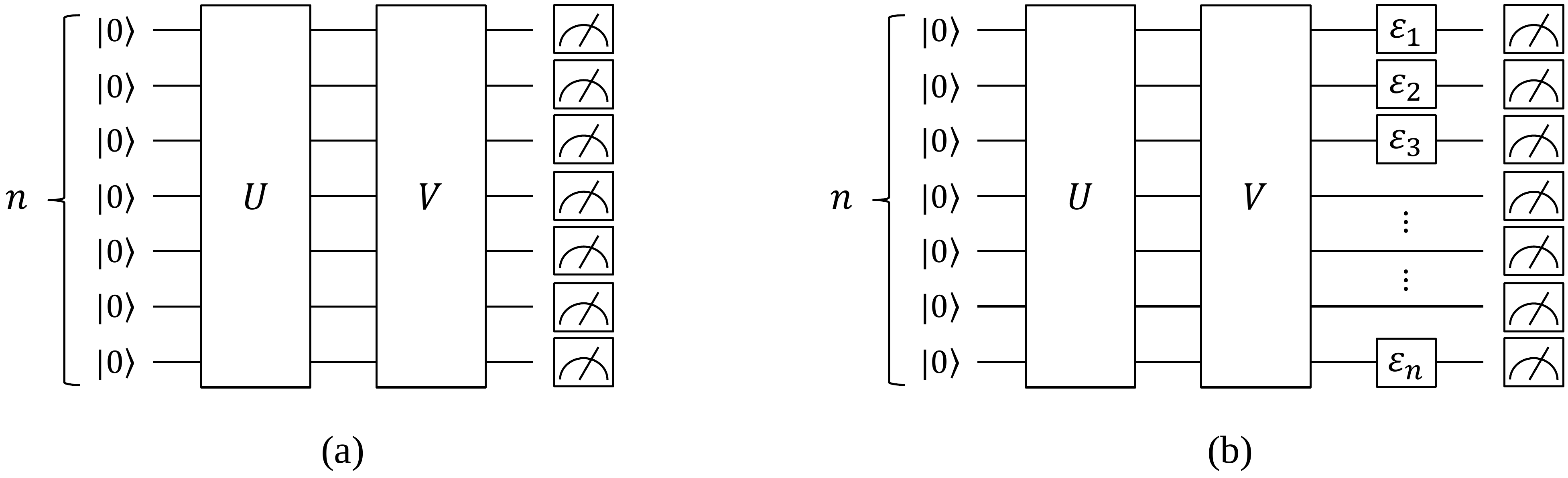}
\caption{(a): CT-ECS circuit $C=VU$, where $U|0^n\rangle$ is CT and $V^\dag Z_j V$ is ECS for any $1\leq j \leq n$. 
After performing $C$, we perform $Z$-basis measurements on all qubits. 
(b): CT-ECS circuit $C=VU$ under noise model~{\bf B}, where $\varepsilon_j$ represents the depolarizing 
channel $D_{\varepsilon_j}$ for any $1\leq j \leq n$.}
\label{figure1}
\end{figure}

Theorem~\ref{basic} assumes noise model~{\bf A} where noise exists only at the end of computation, but, in some cases, 
it can be applied to an input-noise model. For example, Theorem~\ref{basic} holds for IQP circuits when noise model~{\bf A} is 
replaced with a noise model where $D_{\varepsilon}$ is applied to each qubit only at the {\it start} of computation, although 
Bu et al.'s main result implies a similar property of IQP circuits~\cite{Bu}. 
Moreover, Theorem~\ref{general}, which is described below and assumes noise 
model~{\bf B}, also holds for IQP circuits when noise model~{\bf B} is replaced with an input-noise model where the noise rate can 
vary with each qubit. Our main result is similar in spirit to Bu et al.'s, which provides classical algorithms for 
simulating Clifford circuits with nonstabilizer product input states (corresponding to input-noise models). However, we note that, 
in general, it is difficult to relate the output probability distributions of CT-ECS circuits under output-noise models to those of 
Clifford circuits under input-noise models.

Our main focus is on a noise setting, but, from a purely theoretical point of view, it is valuable to analyze the classical 
simulatability of quantum circuits in the noise-free setting. The proof method of Theorem~\ref{basic} is based on computing 
the Fourier coefficients of an output probability distribution, and is useful in the noise-free setting. In fact, it implies that, 
when only $O(\log n)$ qubits are measured, 
any quantum circuit in a class of CT-ECS circuits on $n$ qubits is classically simulatable. More precisely, its output probability 
distribution is classically samplable in polynomial time with polynomial accuracy in $l_1$ norm. This class of CT-ECS 
circuits is defined by a restricted version of ECS operations, and includes IQP, Clifford Magic, conjugated 
Clifford, and constant-depth quantum circuits. It is known that the above property or a similar one holds for 
these quantum circuits (although the notions of approximation vary), but the proofs provided have depended on each 
circuit class~\cite{Terhal,Bremner,Koh,Bouland}. Our analysis 
unifies the previous ones and clarifies a class of quantum circuits for which the above property holds.

Our second result assumes noise model~{\bf B}, which is depicted in Fig.~\ref{figure1}(b). For classically simulating CT-ECS circuits 
with anti-concentrated output probability distributions, we provide a sufficient condition, which is similar to Theorem~\ref{basic}:
\begin{theorem}[informal]\label{general}
Let $C$ be an arbitrary CT-ECS circuit on $n$ qubits such that its output probability 
distribution $p$ satisfies $\sum_{x\in \{0,1\}^n} p(x)^2 \leq \alpha/2^n$ 
for some known constant $\alpha \geq 1$. We assume that a depolarizing 
channel with (possibly unknown) constant rate $0 < \varepsilon_j < 1$ is applied to the $j$-th qubit after performing $C$ for 
any $1 \leq j \leq n$, which yields the probability distribution $\widetilde{p}_{\rm B}$. Moreover, we assume that it is possible to 
choose a constant $\lambda_{\min}$ such that 
$$1 \leq \frac{\varepsilon_{\min}}{\lambda_{\min}} 
\leq 1+c,$$
where $\varepsilon_{\min} = \min \{\varepsilon_j |1 \leq j \leq n\}$ and $c$ is a certain constant depending on $\alpha$, and 
we assume that it is possible to choose a constant $\lambda_j$ such that 
$$0 \leq \varepsilon_j-\lambda_j \leq c\lambda_{\min}$$
for any $1 \leq j \leq n$ with $\varepsilon_j \neq \varepsilon_{\min}$. 
Then, $\widetilde{p}_{\rm B}$ is classically samplable in polynomial time with constant accuracy in $l_1$ norm.
\end{theorem}
\noindent
To the best of our knowledge, this is the 
first analysis of the classical simulatability of quantum circuits under noise model~{\bf B}. Theorem~\ref{general} indicates 
that, under this noise model, if approximate values of the minimum noise rate and the other noise rates are known, 
the presence of small noise drastically affects the classical simulatability of CT-ECS circuits.

\subsection{Overview of Techniques}

To prove Theorem~\ref{basic}, we generalize Bremner et al.'s proof for IQP circuits~\cite{Bremner3}. There 
are two key points. The first one is to provide a general method for approximating the Fourier coefficients of the 
output probability distribution $p$. It is known that the probability distribution $\widetilde{p}_{\rm A}$, which we want to
approximate, can be simply represented by the noise rate $\varepsilon$ and the Fourier 
coefficients $\widehat{p}(s)$ of $p$ for all $s\in\{0,1\}^n$~\cite{Odonnel-book}. We show that, for any CT-ECS circuit 
on $n$ qubits, there exists a polynomial-time classical algorithm for 
approximating each of the low-degree Fourier coefficients of $p$, i.e., $\widehat{p}(s)$ for any $s=s_1\cdots s_n \in \{0,1\}^n$ 
with $\sum_{j=1}^ns_j=O(1)$. Bremner et al.~\cite{Bremner3} showed that such an algorithm exists 
for IQP circuits through a direct calculation of the Fourier coefficients for them. In contrast, we first 
provide a general relation between a quantum circuit and the Fourier coefficients of its output probability distribution. 
We then approximate each of the low-degree Fourier coefficients by combining this general relation with Nest's classical algorithm for 
approximating the inner product value of a particular form defined by a CT state and an ECS operation~\cite{van2}. This 
general relation seems to be a new tool to investigate the output probability distribution of a quantum circuit 
and thus may be of independent interest.

The second key point is to approximate $\widetilde{p}_{\rm A}$ using an approximate value 
of $\varepsilon$. We define a {\it function} $q$ that seems to be close to $\widetilde{p}_{\rm A}$ on the basis 
of its representation with $\varepsilon$ and the Fourier coefficients $\widehat{p}(s)$ for all $s\in\{0,1\}^n$. 
Unfortunately, in contrast to Bremner et al.'s setting, we do not know $\varepsilon$. Thus, using an approximate value $\lambda$ of 
$\varepsilon$, we first choose an appropriate (polynomial) number of 
the low-degree Fourier coefficients used in $q$, and define $q$ based on the above representation of 
$\widetilde{p}_{\rm A}$. More precisely, this number depends on $\lambda$, the constant $\alpha$ associated with 
the anti-concentration assumption, and the desired approximation accuracy. 
We then evaluate the approximation accuracy of $q$. 
Here, we need to care about the error caused by the difference between $\lambda$ and $\varepsilon$, and we 
upper-bound this error using the anti-concentration assumption.

To prove Theorem~\ref{general}, we represent the effect of noise under noise model~{\bf B} as 
the effect of noise under noise model~{\bf A} with rate $\varepsilon_{\min}$ 
and the remaining effects. We do this by transforming the representation of the probability distribution $\widetilde{p}_{\rm B}$, 
which we want to approximate, with several basic 
properties of noise operators on real-valued functions over $\{0,1\}^n$~\cite{Odonnel-book}. The obtained representation 
means that, to sample $\widetilde{p}_{\rm B}$, it suffices to sample the probability 
distribution $\widetilde{p}_{\rm A}$ (resulting from $p$) 
under noise model~{\bf A} with rate $\varepsilon_{\min}$ and then to classically simulate noise corresponding to 
the remaining effects. By Theorem~\ref{basic}, $\widetilde{p}_{\rm A}$ is classically samplable in polynomial time 
with arbitrary constant accuracy in $l_1$ norm. Moreover, we can simulate noise corresponding 
to the remaining effects using the approximate values of $\varepsilon_{\min}$ and $\varepsilon_j$'s, which are not equal 
to $\varepsilon_{\min}$.

\section{Preliminaries}\label{preliminaries}

\subsection{Quantum Circuits and Their Output Probability Distributions}\label{circuit}

Pauli matrices $X$, $Y$, $Z$, and $I$ are
$$
X=
\begin{pmatrix}
0 & 1 \\
1 & 0
\end{pmatrix},\
Y=
\begin{pmatrix}
0 & -i \\
i & 0
\end{pmatrix},\
Z=
\begin{pmatrix}
1 & 0 \\
0 & -1
\end{pmatrix},\
I=
\begin{pmatrix}
1 & 0 \\
0 & 1
\end{pmatrix}.
$$
The Hadamard operation $H$ is defined as $H=(X+Z)/\sqrt{2}$. For any real number $\theta$, 
the rotation operations $R_x(\theta)$ and $R_z(\theta)$ are defined as
$$
R_x(\theta) = \cos\frac{\theta}{2}\,I -i\sin\frac{\theta}{2}\,X,\ 
R_z(\theta) = \cos\frac{\theta}{2}\,I -i\sin\frac{\theta}{2}\,Z.
$$
In particular, $R_z(\pi/4)$ and $R_z(\pi/2)$ are denoted as $T$ and $S$, respectively. It is 
easy to verify that the inverse of $R_x(\theta)$ is $R_x(-\theta)$ and, similarly, the inverse 
of $R_z(\theta)$ is $R_z(-\theta)$. The controlled-$Z$ operation $CZ$ and 
the controlled-controlled-$Z$ operation $CCZ$ are defined as
$$
CZ = |0\rangle\langle 0|\otimes I + |1\rangle\langle 1|\otimes Z,\ 
CCZ = |00\rangle\langle 00|\otimes I + |01\rangle\langle 01|\otimes I +
|10\rangle\langle 10|\otimes I + |11\rangle\langle 11|\otimes Z,
$$
where the states $|0\rangle$ and $|1\rangle$ are $\pm 1$-eigenstates of $Z$, respectively. 
The operations $H$, $S$, and $CZ$ are called Clifford operations. 

A quantum circuit consists of elementary gates, each of which is in the gate set $\mathcal G$. Here, 
$\mathcal G=\{R_x(\theta),R_z(\theta),CZ|\theta =\pm 2\pi /2^t \ {\rm with \ integer} \ t\geq 1\}$. 
Each $R_x(\theta)$ has its inverse in $\mathcal G$ and so does $R_z(\theta)$. Moreover, each of 
the one-qubit operations $X$, $Y$, $Z$, $I$, and $H$ can be decomposed into a constant 
number of $R_x(\theta)$'s and $R_z(\theta)$'s (up to an unimportant global phase). 
Similarly, $CCZ$ can be decomposed into a constant number of $R_x(\theta)$'s, $R_z(\theta)$'s, 
and $CZ$'s. Since the gate set $\{H,T,CZ\}$ is 
approximately universal for quantum computation~\cite{Nielsen}, so is $\mathcal G$. 
The complexity measures of a quantum circuit are its size and depth. The size of a 
quantum circuit is the number of elementary gates in the circuit. To define the depth, we 
regard the circuit as a set of layers $1,\ldots,d$ consisting of elementary gates, where gates 
in the same layer act on pairwise disjoint sets of qubits and any gate in layer $j$ is applied 
before any gate in layer $j+1$. The depth is defined as the smallest possible value 
of $d$~\cite{Fenner}.

We deal with a (polynomial-time) uniform family of polynomial-size quantum 
circuits $\{C_n\}_{n\geq 1}$. Each $C_n$ has $n$ input qubits initialized to $|0^n\rangle$. 
After performing $C_n$, we perform  $Z$-basis measurements on all qubits. 
The output probability distribution $p$ of $C_n$ over $\{0,1\}^n$ is defined as $p(x)=|\langle x|C_n|0^n\rangle|^2$. 
A symbol denoting a quantum circuit also denotes its matrix representation in the $Z$ basis. 
The (polynomial-time) uniformity means that the function $1^n \mapsto \overline{C_n}$ is computable by a polynomial-time 
classical Turing machine, where $\overline{C_n}$ is the classical description of $C_n$~\cite{Nishimura-tcs}.

\subsection{Fourier Expansions and Effects of Noise}\label{fourier}

Let $f:\{0,1\}^n \to {\mathbb R}$ be an arbitrary (real-valued) function. Then, $f$ can be 
uniquely represented as an ${\mathbb R}$-linear combination of $2^n$ basis functions
$$f(x)=\sum_{s \in\{0,1\}^n}\widehat{f}(s)(-1)^{s\cdot x},$$
which is called the Fourier expansion of $f$~\cite{Odonnel-book}. Here, $\widehat{f}(s)$ is called the Fourier 
coefficient of $f$ and the symbol ``$\cdot$'' represents the inner product of two $n$-bit strings, i.e., 
$s\cdot x =\sum_{j=1}^n s_jx_j$ for any $s=s_1\cdots s_n$, $x=x_1\cdots x_n\in \{0,1\}^n$. It holds that
$$\widehat{f}(s)=\frac{1}{2^n}\sum_{x \in\{0,1\}^n}f(x)(-1)^{s\cdot x}$$
for any $s\in\{0,1\}^n$. The $l_k$ norm of $f$ is defined as
$$||f||_k = \left(\sum_{x \in\{0,1\}^n}|f(x)|^k\right)^{1/k},$$
where $k=1,2$. It is known that $||f||_1^2 \leq 2^n||f||_2^2 = 2^{2n}||\widehat{f}||_2^2$~\cite{Odonnel-book}. 

Let $C$ be an arbitrary quantum circuit on $n$ qubits and $p$ be its output 
probability distribution over $\{0,1\}^n$. We consider $C$ under noise model~{\bf A} 
where a depolarizing channel $D_\varepsilon$ with rate $0 < \varepsilon < 1$ is applied to each qubit 
after performing $C$. Here, $D_\varepsilon(\rho) = (1-\varepsilon)\rho + \varepsilon\frac{I}{2}$ for any density 
operator $\rho$ of a qubit. We perform $Z$-basis measurements on all qubits and let $\widetilde{p}_{\rm A}$ be the 
resulting probability distribution over $\{0,1\}^n$. 
As shown in~\cite{Bremner3}, we can sample $\widetilde{p}_{\rm A}$ by sampling an $n$-bit string 
according to $p$ and then flipping each bit of the string with probability $\varepsilon/2$. This implies the following Fourier 
expansion of $\widetilde{p}_{\rm A}$~\cite{Odonnel-book}:
$$\widetilde{p}_{\rm A}(x)=\sum_{s \in\{0,1\}^n}(1-\varepsilon)^{|s|}\widehat{p}(s)(-1)^{s\cdot x},$$
where $|s|=\sum_{j=1}^ns_j$ for any $s=s_1\cdots s_n \in \{0,1\}^n$. We also consider $C$ under noise model~{\bf B} 
where $D_{\varepsilon_j}$ with rate $0 < \varepsilon_j < 1$ is applied to the $j$-th qubit after performing $C$ for any $1 \leq j \leq n$. 
We perform $Z$-basis measurements on all qubits and let $\widetilde{p}_{\rm B}$ be the resulting probability distribution over $\{0,1\}^n$. 
As with $\widetilde{p}_{\rm A}$, we can sample $\widetilde{p}_{\rm B}$ by sampling an $n$-bit string 
according to $p$ and then flipping its $j$-th bit with probability $\varepsilon_j/2$ for any $1 \leq j \leq n$. 
The Fourier expansion of $\widetilde{p}_{\rm B}$ is as follows~\cite{Odonnel-book}:
$$\widetilde{p}_{\rm B}(x)=\sum_{s \in\{0,1\}^n}\left[\prod_{j=1}^n  (1-\varepsilon_j)^{s_j}\right]\widehat{p}(s)(-1)^{s\cdot x}.$$

\subsection{CT States and ECS Operations}\label{idea}

We introduce CT states and (a restricted version of) ECS operations~\cite{van2}. Let $|\varphi\rangle$ be an arbitrary (pure) quantum 
state on $n$ qubits and $p$ be the probability distribution over $\{0,1\}^n$ defined as $p(x)=|\langle x |\varphi\rangle|^2$. Then, 
$|\varphi\rangle$ is CT if $p$ is classically samplable in polynomial time and, for any $x \in\{0,1\}^n$, 
$\langle x |\varphi\rangle$ is classically computable in polynomial time.\footnote{For simplicity, we require perfect accuracy 
in sampling probability distributions and computing values, although irrational numbers may be involved. Precisely 
speaking, it suffices to require exponential accuracy. This is also applied to similar situations in this paper, 
such as the definition of ECS operations.} An example of a CT state is a product state.

Let $U$ be an arbitrary quantum operation on $n$ qubits that 
is both unitary and Hermitian. The operation $U$ is sparse if there exists a polynomial $s$ in $n$ such that, 
for any $x \in \{0,1\}^n$, $U|x\rangle$ is a linear combination of at 
most $s(n)$ computational basis states. When $U$ is a sparse operation (associated with $s$), 
for any $1 \leq j \leq s(n)$, we define two functions, $\beta_j:\{0,1\}^n \to {\mathbb C}$ and 
$\gamma_j:\{0,1\}^n \to \{0,1\}^n$,  as follows: for any $x \in \{0,1\}^n$, 
if the $j$-th non-zero entry exists in the column indexed by $x$ when traversing this column 
from top to bottom,  $\beta_j(x)$ is this entry and $\gamma_j(x)$ is the row index 
associated with $\beta_j(x)$. If the $j$-th non-zero entry does not exist in this column, 
$\beta_j(x)=0$ and $\gamma_j(x)=0^n$. The sparse operation $U$ is ECS if, for any $x \in \{0,1\}^n$ and $1 \leq j \leq s(n)$, 
$\beta_j(x)$ and $\gamma_j(x)$ are classically computable in polynomial time. In 
particular, an ECS operation with $s(n)=O(1)$ is called ECS$_1$. Moreover, an ECS operation with $s(n)=1$ is 
called efficiently computable basis-preserving.

An efficiently computable basis-preserving operation preserves the class of CT states~\cite{van2}:
\begin{theorem}[\cite{van2}]\label{van2-1}
Let $|\varphi\rangle$ be an arbitrary CT state on $n$ qubits and $U$ be an arbitrary efficiently computable 
basis-preserving operation on $n$ qubits. Then, $U|\varphi\rangle$ is CT.
\end{theorem}
\noindent
The following theorem is a rephrased version of the one in~\cite{van2}:
\begin{theorem}[\cite{van2}]\label{van2-2}
Let $U$ be an arbitrary quantum operation on $n$ qubits such that $U|0^n\rangle$ is CT, and $O$ be an arbitrary 
observable with $||O|| \leq 1$, where $||O||$ is the absolute value of the largest eigenvalue 
of $O$. Let $V$ be an arbitrary quantum operation on $n$ qubits such that $V^\dag OV$ is ECS, and $f$ be an arbitrary 
polynomial in $n$. Then, there exists a polynomial-time randomized algorithm which outputs a real number $r$ such that 
$${\rm Pr}\left[\left|\langle 0^n|U^\dag V^\dag O VU|0^n\rangle -r\right| \leq \frac{1}{f(n)}\right] \geq 1-\frac{1}{\exp(n)}.$$
\end{theorem}

\section{Target Quantum Circuits and Associated Fourier Coefficients}\label{parallel}

\subsection{CT-ECS Circuits}\label{parallel1}

We focus on a new class of quantum circuits defined by CT states and ECS operations:
\begin{definition}
A quantum circuit $C$ on $n$ qubits initialized to $|0^n\rangle$ is CT-ECS if $C$ consists of two blocks $C=VU$ such 
that $U$ and $V$ are polynomial-size quantum circuits, $U|0^n\rangle$ is CT, and $V^\dag Z_j V$ is ECS for 
any $1\leq j \leq n$, where $Z_j$ is a Pauli-$Z$ operation on the $j$-th qubit.
\end{definition}
\noindent
To provide examples of CT-ECS circuits, we define IQP, Clifford Magic, conjugated Clifford, 
and constant-depth quantum circuits on $n$ qubits as follows:
\begin{itemize}
\item An IQP circuit is of the form $H^{\otimes n}DH^{\otimes n}$, where $D$ is a 
polynomial-size quantum circuit consisting of $Z$, $CZ$, and $CCZ$ gates~\cite{Bremner2}.

\item A Clifford Magic circuit is of the form  $ET^{\otimes n}H^{\otimes n}$, where $E$ is a 
polynomial-size Clifford circuit, which consists of $H$, $S$, and $CZ$ gates~\cite{Yoganathan}.

\item A conjugated Clifford circuit is of the form 
$R_x(-\theta)^{\otimes n} R_z(-\phi)^{\otimes n} E R_z(\phi)^{\otimes n}R_x(\theta)^{\otimes n}$ 
for arbitrary real numbers $\phi,\theta$, where $E$ is a polynomial-size Clifford circuit~\cite{Bouland}.

\item A constant-depth quantum circuit is a polynomial-size quantum circuit $F$ whose depth is 
constant~\cite{Terhal,Bremner,Juan}.
\end{itemize}
In the common definition of an IQP circuit, $D$ consists of more general $Z$-diagonal gates. However, for simplicity, we adopt 
the above definition. The resulting class includes a quantum circuit of the form $H^{\otimes n}DH^{\otimes n}$ 
such that its output probability distribution is anti-concentrated and not classically simulatable (under plausible assumptions).

We show that the above circuits are CT-ECS:
\begin{lemma}\label{four}
Let $C$ be one of the following quantum circuits on $n$ qubits: an IQP, a Clifford 
Magic, a conjugated Clifford, or a constant-depth quantum circuit. 
Then, $C$ is CT-ECS.
\end{lemma}
\begin{proof}
When $C$ is an IQP circuit, $C=H^{\otimes n}DH^{\otimes n}$, where $D$ is 
a polynomial-size quantum circuit consisting of $Z$, $CZ$, and $CCZ$ gates. We consider $U=DH^{\otimes n}$ 
and $V=H^{\otimes n}$. Since $H^{\otimes n}|0^n\rangle$ is 
a product state and $D$ is an efficiently computable basis-preserving operation, by Theorem~\ref{van2-1}, 
$U|0^n\rangle=DH^{\otimes n}|0^n\rangle$ is CT. Moreover, $V^\dag Z_jV=X_j$, which is obviously 
ECS (in fact, efficiently computable basis-preserving), where $X_j$ is a Pauli-$X$ operation on the $j$-th qubit. 
Thus, $C$ is CT-ECS.

When $C$ is a Clifford Magic circuit, $C=ET^{\otimes n}H^{\otimes n}$, where $E$ is a polynomial-size Clifford circuit. 
We consider $U=T^{\otimes n}H^{\otimes n}$ and $V=E$. Since $U|0^n\rangle=T^{\otimes n}H^{\otimes n}|0^n\rangle$ is 
a product state, it is CT. Moreover, since $E$ is a Clifford circuit, $V^\dag Z_jV=E^\dag Z_jE$ is a Pauli 
operation on $n$ qubits, which is obviously ECS (in fact, efficiently computable basis-preserving). Thus, $C$ is CT-ECS.

When $C$ is a conjugated Clifford circuit, 
$C=R_x(-\theta)^{\otimes n} R_z(-\phi)^{\otimes n} E R_z(\phi)^{\otimes n}R_x(\theta)^{\otimes n}$ 
for arbitrary real numbers $\phi,\theta$, where $E$ is a polynomial-size Clifford circuit. We 
consider $U=R_z(\phi)^{\otimes n}R_x(\theta)^{\otimes n}$ and 
$V=R_x(-\theta)^{\otimes n} R_z(-\phi)^{\otimes n} E$. 
Since $U|0^n\rangle=R_z(\phi)^{\otimes n}R_x(\theta)^{\otimes n}|0^n\rangle$ is a product state, it is CT. Moreover, 
$$V^\dag Z_jV= E^\dag R_z(\phi)_jR_x(\theta)_j Z_j R_x(-\theta)_j R_z(-\phi)_j E.$$
The coefficients $\alpha_1,\alpha_2,\alpha_3,\alpha_4\in {\mathbb C}$ satisfying the following relation are classically computable 
in constant time:
$$\alpha_1 I+\alpha_2 Z+ \alpha_3 X +\alpha_4 Y=R_z(\phi)_jR_x(\theta)_j Z_j R_x(-\theta)_j R_z(-\phi)_j.$$
This implies that $V^\dag Z_jV = \alpha_1 I+\alpha_2  E^\dag Z E+ \alpha_3  E^\dag X E+\alpha_4  E^\dag Y E$. 
Since $E$ is a Clifford circuit, the operations $E^\dag Z E$, $E^\dag X E$, and $E^\dag Y E$ are Pauli operations 
on $n$ qubits, which implies that $V^\dag Z_jV$ is ECS. Thus, $C$ is CT-ECS.

When $C$ is a constant-depth quantum circuit, we consider $U=I$ and $V=C$. Of course, 
$U|0^n\rangle=|0^n\rangle$ is CT. Moreover, since each elementary gate in this paper acts only on a constant number of qubits, 
$V^\dag Z_jV$ also acts only on a constant number of qubits, 
which implies that $V^\dag Z_jV$ is ECS. Thus, $C$ is CT-ECS.
\end{proof}
\noindent
Lemma~\ref{four} implies that there exist various CT-ECS circuits such that their output probability distributions are 
anti-concentrated and not classically simulatable (under plausible assumptions), 
although we do not know whether the output probability distributions of constant-depth quantum circuits are anti-concentrated.

The above proof implies that, for any quantum circuit $C$ in Lemma~\ref{four}, the associated ECS operation $V^\dag Z_jV$ 
satisfies the condition where, for any $x\in\{0,1\}^n$, $V^\dag Z_j V|x\rangle$ can be represented as a linear combination of at 
most $O(1)$ computational basis states. In other words, $V^\dag Z_jV$ is ECS$_1$. This defines a class of CT-ECS circuits, 
whose elements we call CT-ECS$_1$ circuits. In Section~\ref{applications}, we consider the classical 
simulatability of CT-ECS$_1$ circuits on $n$ qubits in the noise-free setting when only $O(\log n)$ qubits are measured.

\subsection{Approximating the Associated Fourier Coefficients}\label{coefficient}

We provide a general relation between a quantum circuit and the Fourier coefficients of its output probability 
distribution:
\begin{lemma}\label{fourier-rep}
Let $C$ be an arbitrary quantum circuit on $n$ qubits initialized to $|0^n\rangle$ and $p$ be its output 
probability distribution over $\{0,1\}^n$. Then, 
$$\widehat{p}(s)=\frac{1}{2^n}\langle 0^n|C^\dag Z^s C|0^n\rangle$$
for any $s=s_1\cdots s_n \in\{0,1\}^n$, 
where $Z^s=\bigotimes_{j=1}^n Z_j^{s_j}$, i.e., the tensor product of a Pauli-$Z$ operation on 
the $j$-th qubit with $s_j=1$ for any $1 \leq j \leq n$.
\end{lemma}
\begin{proof}
We transform the representation of the Fourier coefficient described in Section~\ref{fourier} as follows:
\begin{align}\label{fourier1}
\widehat{p}(s) 
& = \frac{1}{2^n}\sum_{x \in\{0,1\}^n}p(x)(-1)^{s\cdot x} 
= \frac{1}{2^n}\sum_{x \in\{0,1\}^n}\langle 0^n|C^\dag |x\rangle\langle x| 
C|0^n\rangle(-1)^{s\cdot x}\nonumber
\\
& =  \frac{1}{2^n}\sum_{x \in\{0,1\}^n}\langle 0^n|C^\dag 
X^x|0^n\rangle\langle 0^n |X^x C|0^n\rangle(-1)^{s\cdot x},
\end{align}
where $X^x$ denotes $H^{\otimes n} Z^x H^{\otimes n}$ for any $x \in\{0,1\}^n$. Since it holds that 
$$|0^n\rangle\langle 0^n| =\frac{1}{2^n}\sum_{t \in\{0,1\}^n}Z^t,$$
the above representation (\ref{fourier1}) of $\widehat{p}(s)$ implies that
\begin{align*} 
\widehat{p}(s)
& = \frac{1}{2^{2n}}\sum_{x \in\{0,1\}^n}\sum_{t \in\{0,1\}^n}\langle 0^n|C^\dag X^xZ^tX^x C|0^n\rangle(-1)^{s\cdot x} \\
& = \frac{1}{2^{2n}}\sum_{t \in\{0,1\}^n}\langle 0^n|C^\dag Z^t C|0^n\rangle\sum_{x \in\{0,1\}^n}(-1)^{(s+t)\cdot x} 
= \frac{1}{2^n}\langle 0^n|C^\dag Z^s C|0^n\rangle,              
\end{align*}
where $s+t$ is the bit-wise addition of $s$ and $t$ modulo 2. This is the desired representation.
\end{proof}

Using Theorem~\ref{van2-2} and Lemma~\ref{fourier-rep}, we show that the low-degree Fourier coefficients 
of the output probability distribution of a CT-ECS circuit can be approximated classically in polynomial time:
\begin{lemma}\label{appro-fourier}
Let $C$ be an arbitrary CT-ECS circuit on $n$ qubits and $p$ be its output probability distribution over $\{0,1\}^n$. 
Let $f$ be an arbitrary polynomial in $n$ and $s$ be an arbitrary element of $\{0,1\}^n$ with $|s|=O(1)$. 
Then, there exists a polynomial-time randomized algorithm which outputs a real number $\widehat{p}'(s)$ such that
$${\rm Pr}\left[|\widehat{p}(s) - \widehat{p}'(s)| \leq \frac{1}{2^n f(n)}\right] \geq 1-\frac{1}{\exp(n)}.$$
\end{lemma}
\begin{proof}
Since $C$ is CT-ECS, it can be represented as $C=VU$ such that $U|0^n\rangle$ is CT and 
$V^\dag Z_j V$ is ECS for any $1\leq j \leq n$. Let $s$ be an arbitrary element of $\{0,1\}^n$ with $|s|=O(1)$, and we 
assume that $s_{j_1} = \cdots =s_{j_{|s|}}=1$. In this case,
$$V^\dag Z^s V = \prod_{k=1}^{|s|}(V^\dag Z_{j_k} V).$$
Since $V^\dag Z_{j_k} V$ is ECS, $V^\dag Z^s V$ is the product of a constant number of ECS operations. 
A simple calculation shows that such a product is also ECS. Thus, $V^\dag Z^s V$ is ECS.

Since $U|0^n\rangle$ is CT and $||Z^s|| \leq 1$, by Theorem~\ref{van2-2}, there exists a polynomial-time randomized algorithm 
which outputs a real number $r(s)$ such that
$${\rm Pr}\left[\left|\langle 0^n|U^\dag V^\dag Z^s VU|0^n\rangle -r(s)\right| \leq \frac{1}{f(n)}\right] \geq 1-\frac{1}{\exp(n)}.$$
By Lemma~\ref{fourier-rep} with $C=VU$,
$$\widehat{p}(s)=\frac{1}{2^n}\langle 0^n|U^\dag V^\dag Z^s VU|0^n\rangle$$
and thus the desired relation holds by defining $\widehat{p}'(s) = r(s)/2^n$.
\end{proof}

For any probability distribution $p$ over $\{0,1\}^n$, it holds that
$$\widehat{p}(0^n)=\frac{1}{2^n}\sum_{x\in \{0,1\}^n} p(x)=\frac{1}{2^n}.$$ 
Thus, when we consider a classical algorithm for approximating $\widehat{p}(s)$, we only consider an 
algorithm that outputs $1/2^n$ when $s=0^n$. This slightly simplifies the analysis 
of the classical simulatability of CT-ECS circuits in the following sections.

\section{CT-ECS Circuits under Noise Model~{\bf A}}

In this section, we prove Theorem~\ref{basic}. Its precise statement is as follows:
\setcounter{theorem}{0}
\begin{theorem}
Let $C$ be an arbitrary CT-ECS circuit on $n$ qubits such that its output probability 
distribution $p$ over $\{0,1\}^n$ is anti-concentrated, i.e., $\sum_{x\in \{0,1\}^n} p(x)^2 \leq \alpha/2^n$ 
for some known constant $\alpha \geq 1$. Let $0 < \delta < 1$ be an arbitrary constant. We assume that
\begin{itemize}
\item a depolarizing channel with (possibly unknown) constant rate $0 < \varepsilon < 1$ is applied to each 
qubit after performing $C$, which yields the probability distribution $\widetilde{p}_{\rm A}$ over $\{0,1\}^n$, and
\item it is possible to choose a constant $\lambda$ such that 
$$1 \leq \frac{\varepsilon}{\lambda} \leq 1+\frac{1}{\frac{10\sqrt{\alpha}}{\delta}\log\frac{10\sqrt{\alpha}}{\delta}}.$$ 
\end{itemize}
Then, there exists a polynomial-time randomized algorithm which outputs (a classical 
description of) a probability distribution $\widetilde{q}_{\rm A}$ over $\{0,1\}^n$ such that 
$${\rm Pr}\left[||\widetilde{p}_{\rm A}-\widetilde{q}_{\rm A}||_1 \leq \delta\right] \geq 1-\frac{1}{\exp(n)}$$ 
and $\widetilde{q}_{\rm A}$ is classically samplable in polynomial time.
\end{theorem}
\noindent
First, we define a function over $\{0,1\}^n$ that is close to 
$\widetilde{p}_{\rm A}$, which we want to approximate. This is done by using the approximate value $\lambda$ of the 
noise rate $\varepsilon$ and the approximate values $\widehat{p}'(s)$ of the low-degree Fourier coefficients of 
$p$ obtained by Lemma~\ref{appro-fourier}. Then, we sample a probability distribution close to the function.

\subsection{Function Close to the Target Probability Distribution}

We show that the probability distribution $\widetilde{p}_{\rm A}$ can be approximated by a 
function whose Fourier coefficients can be obtained classically in polynomial time:

\begin{lemma}\label{coefficients}
Let $C$ be an arbitrary CT-ECS circuit on $n$ qubits such that its output probability distribution $p$ 
satisfies $\sum_{x\in \{0,1\}^n} p(x)^2 \leq \alpha/2^n$ 
for some known constant $\alpha \geq 1$. Let $0 < \delta < 1$ be an arbitrary constant. We assume that
\begin{itemize}
\item a depolarizing channel with constant rate $0 < \varepsilon < 1$ is applied to each qubit 
after performing $C$, which yields the probability distribution $\widetilde{p}_{\rm A}$, and
\item it is possible to choose a constant $\lambda$ such that 
$$1 \leq \frac{\varepsilon}{\lambda} \leq 1+ \frac{1}{\frac{10\sqrt{\alpha}}{\delta}\log\frac{10\sqrt{\alpha}}{\delta}}.$$ 
\end{itemize}
Then, there exists a polynomial-time randomized algorithm which outputs the Fourier 
coefficients of a function $q$ over $\{0,1\}^n$ such that 
$${\rm Pr}\left[ ||\widetilde{p}_{\rm A}-q||_1 \leq \frac{\delta}{3}\right] \geq 1-\frac{1}{\exp(n)}.$$
\end{lemma}
\begin{proof}
As described in Section~\ref{fourier}, the probability distribution $\widetilde{p}_{\rm A}$ is represented as
$$\widetilde{p}_{\rm A}(x)=\sum_{s \in\{0,1\}^n}(1-\varepsilon)^{|s|}\widehat{p}(s)(-1)^{s\cdot x}.$$
Using the known constants $\alpha$, $\delta$, and $\lambda$, we fix an integer constant
$$c = \left\lceil\frac{1}{\lambda}\log \frac{10\sqrt{\alpha}}{\delta}\right\rceil.$$
Since $ 0 < \lambda < 1$ and $10\sqrt{\alpha}/\delta > 10$, $c >3$. The definition of $c$ implies that
$$\frac{1}{\lambda}\log \frac{10\sqrt{\alpha}}{\delta} \leq c \leq \frac{1}{\lambda}\log \frac{10\sqrt{\alpha}}{\delta} +1 
\leq \frac{2}{\lambda}\log \frac{10\sqrt{\alpha}}{\delta}.$$ 
Thus, $\alpha/2^{2\lambda c} \leq \delta^2/100$. Moreover, 
$0\leq c(\varepsilon -\lambda) \leq \frac{\delta}{5\sqrt{\alpha}}$ since
$$
\lambda \leq \varepsilon \leq \lambda +\frac{\lambda}{\frac{10\sqrt{\alpha}}{\delta}\log\frac{10\sqrt{\alpha}}{\delta}}
\leq \lambda + \frac{\lambda}{\frac{10\sqrt{\alpha}}{\delta}\cdot\frac{c\lambda}{2}} \leq \lambda + \frac{\delta}{5c\sqrt{\alpha}}.
$$

By Lemma~\ref{appro-fourier} with $f(n)=10(n^c+1)/\delta$ and an arbitrary $s\in\{0,1\}^n$ 
with $|s| \leq c$, there exists a polynomial-time randomized algorithm which outputs $\widehat{p}'(s)$ such that
$${\rm Pr}\left[|\widehat{p}(s) - \widehat{p}'(s)| \leq \frac{1}{2^n f(n)}\right] \geq 1-\frac{1}{\exp(n)}.$$
We compute $\widehat{p}'(s)$ for all $s \in \{0,1\}^n\setminus\{0^n\}$ with $|s| \leq c$, and define 
$\widehat{p}'(0^n)=1/2^n$ (as described at the end of Section~\ref{coefficient}) and $\widehat{p}'(s)=0$ 
for all $s \in \{0,1\}^n$ with $|s| > c$. 
Since $|\{s \in \{0,1\}^n ||s|\leq c\}|\leq n^c +1$, it takes polynomial time to compute all these values. Moreover, 
$${\rm Pr}\left[\forall s \in \{0,1\}^n \ {\rm with} \ |s| \leq c,\ |\widehat{p}(s) - \widehat{p}'(s)| \leq \frac{1}{2^n f(n)}\right] 
\geq 1-\frac{1}{\exp(n)}.$$
This can be simply shown by a direct application of the inequality $(1-a)^r \geq 1-ra$ for an arbitrary 
real number $0 \leq a \leq 1$ and integer $r \geq 1$.

We define a function $q$ over $\{0,1\}^n$ as
$$q(x) = \sum_{s\in\{0,1\}^n, |s| \leq c}(1-\lambda)^{|s|}\widehat{p}'(s)(-1)^{s\cdot x}.$$
In the following, we show that $||\widetilde{p}_{\rm A} - q||_1 \leq \delta/3$ 
under the assumption that, for all $s \in \{0,1\}^n$ with $|s| \leq c$, 
$|\widehat{p}(s) - \widehat{p}'(s)| \leq \frac{1}{2^n f(n)}$. 
A direct calculation with the relations described in Section~\ref{fourier} shows that
\begin{align}\label{difference}
\hspace{-.1cm}
||\widetilde{p}_{\rm A} - q||_1^2 & \leq 2^n||\widetilde{p}_{\rm A} - q||_2^2 = 2^{2n}||\widehat{\widetilde{p}_{\rm A}} - \widehat{q}||_2^2
\nonumber
\\
& =
2^{2n}\sum_{s: |s| > c}(1-\varepsilon)^{2|s|}\widehat{p}(s)^2
+
2^{2n}\sum_{s:|s| \leq c}
\left[(1-\varepsilon)^{|s|}\widehat{p}(s)- (1-\lambda)^{|s|}\widehat{p}'(s)  \right]^2.
\end{align}
Using the bounds $\sum_{x\in \{0,1\}^n} p(x)^2 \leq \alpha/2^n$ and $\alpha/2^{2\lambda c} \leq \delta^2/100$, 
we upper-bound the first term of (\ref{difference}) as follows:
\begin{align*}
2^{2n}\sum_{s: |s| > c}(1-\varepsilon)^{2|s|}\widehat{p}(s)^2
& \leq
2^{2n}(1-\lambda)^{2c}\sum_{s: |s| > c}\widehat{p}(s)^2
\leq
2^{2n}(1-\lambda)^{2c}\sum_{s\in \{0,1\}^n}\widehat{p}(s)^2 \\
& =
2^{n}(1-\lambda)^{2c}\sum_{x\in \{0,1\}^n}p(x)^2
\leq \frac{2^n}{e^{2\lambda c}}\cdot\frac{\alpha}{2^n} \leq \frac{\alpha}{2^{2\lambda c}} \leq \frac{\delta^2}{100}.
\end{align*}
Then, we upper-bound the second term of (\ref{difference}) as follows: 
\begin{align}\label{intermediate}
& 2^{2n}\sum_{s: |s| \leq c}
\left[(1-\varepsilon)^{|s|}\widehat{p}(s)- (1-\lambda)^{|s|}\widehat{p}'(s) \right]^2\nonumber\\
& =
2^{2n}\sum_{s: |s| \leq c}
\left\{
\left[(1-\varepsilon)^{|s|}-(1-\lambda)^{|s|}\right]\widehat{p}(s) 
+(1-\lambda)^{|s|}(\widehat{p}(s)-\widehat{p}'(s))\right\}^2\nonumber\\
& \leq
2^{2n}\sum_{s: |s| \leq c}
\left\{
\left[(1-\lambda)^{|s|}-(1-\varepsilon)^{|s|}\right]|\widehat{p}(s)| 
+(1-\lambda)^{|s|}|\widehat{p}(s)-\widehat{p}'(s)|\right\}^2\nonumber\\
& \leq
2^{2n}\sum_{s: |s| \leq c}\left(c(\varepsilon-\lambda)|\widehat{p}(s)| + \frac{1}{2^nf(n)}  \right)^2,
\end{align}
where the last inequality is due to the fact 
that $(1-\lambda)^{|s|}-(1-\varepsilon)^{|s|} \leq |s|(\varepsilon-\lambda) \leq c(\varepsilon-\lambda)$ and 
$|\widehat{p}(s) - \widehat{p}'(s)| \leq \frac{1}{2^n f(n)}$, where $f(n)=10(n^c+1)/\delta$.

Using the bounds $0\leq c(\varepsilon -\lambda) \leq \frac{\delta}{5\sqrt{\alpha}}$,
$|\widehat{p}(s)| \leq 1/2^n$, $|\{s \in \{0,1\}^n ||s|\leq c\}|\leq n^c +1$, and 
$\sum_{x\in \{0,1\}^n} p(x)^2 \leq \alpha/2^n$, we upper-bound the value~(\ref{intermediate}) as follows:
\begin{align*}
& 2^{2n}\sum_{s: |s| \leq c}
\left(c(\varepsilon-\lambda)|\widehat{p}(s)| + \frac{1}{2^nf(n)}  \right)^2 
\leq
2^{2n}\sum_{s: |s| \leq c}
\left(\frac{\delta|\widehat{p}(s)|}{5\sqrt{\alpha}} + \frac{1}{2^nf(n)}  \right)^2 \\
& =
2^{2n}\sum_{s: |s| \leq c}
\left(\frac{\delta^2\widehat{p}(s)^2}{25\alpha} + \frac{2\delta|\widehat{p}(s)|}{5\sqrt{\alpha}2^nf(n)}
+ \frac{1}{2^{2n}f(n)^2} \right)\\
& \leq
\frac{2^{n}\delta^2}{25\alpha}\sum_{x \in \{0,1\}^n}{p}(x)^2
+
\frac{2\delta (n^c+1)}{5\sqrt{\alpha}f(n)}
+
\frac{n^c+1}{f(n)^2}
\leq 
\frac{\delta^2}{25}  + \frac{\delta^2}{25} +\frac{\delta^2}{100}= \frac{9\delta^2}{100}.
\end{align*}
Combining the upper bounds of the above two terms implies that
$$||\widetilde{p}_{\rm A} - q||_1^2 \leq \frac{\delta^2}{100} + \frac{9\delta^2}{100}  \leq \frac{\delta^2}{9},$$
which is the desired bound.
\end{proof}

\subsection{Sampling a Probability Distribution Close to the Function}\label{sampling}

The remaining problem is to sample a probability distribution close to the function $q$ defined in the proof of Lemma~\ref{coefficients}. 
To do this, we use a classical sampling algorithm proposed by Bremner et al.~\cite{Bremner3}, although the following analysis is slightly 
simpler than the previous one. We represent $q$ as 
$$q(x) = \sum_{s\in\{0,1\}^n, |s| \leq c}\widehat{q}(s)(-1)^{s\cdot x},$$
where $\widehat{q}(s)=(1-\lambda)^{|s|}\widehat{p}'(s)$. It holds that 
$\sum_{x\in\{0,1\}^n}q(x)=2^n\widehat{q}(0^n)=2^n\widehat{p}'(0^n)=1$.

We define $S_\epsilon =1$ and
$$S_y=\sum_{x\in\{0,1\}^n,x_1\cdots x_k=y}q(x)$$
for any $1 \leq k \leq n$ and $y \in \{0,1\}^k$, where $\epsilon$ denotes the empty string. The sampling algorithm 
is described as follows:
\begin{enumerate}
\item Set $y \gets \epsilon$.

\item Perform the following procedure $n$ times:
\begin{enumerate}
\item If $S_{yz} <0$ for some $z \in \{0,1\}$, set $y \gets y\bar{z}$, where $\bar{z}=1-z$.

\item Otherwise, set $y \gets y0$ with probability $S_{y0}/S_y$ and $y \gets y1$ with probability $1-S_{y0}/S_y$.
\end{enumerate}

\item Output $y \in\{0,1\}^n$.
\end{enumerate}
A direct calculation shows that
$$S_y=2^{n-k}\sum_{s\in\{0,1\}^n, |s| \leq c,s_{k+1}\cdots s_n=0^{n-k}}\widehat{q}(s)(-1)^{s_1\cdots s_k \cdot y}$$
for any $1 \leq k \leq n$ and $y \in \{0,1\}^k$. 
Since $S_y$ is classically computable in polynomial time (when we have $\widehat{q}(s)$ for all 
$s\in\{0,1\}^n$ with $|s| \leq c$), the runtime of the sampling algorithm is also polynomial. This sampling algorithm 
defines a real-valued function, denoted as ${\rm Alg}(q)$, that maps $y \in\{0,1\}^n$ to the 
probability that the sampling algorithm outputs $y$.

To make the analysis of ${\rm Alg}(q)$ easier, Bremner et al.~\cite{Bremner3} defined another function, 
${\rm Fix}(q)=\left(\sum_{x\in\{0,1\}^n}q(x)\right){\rm Alg}(q)$, 
and showed that ${\rm Fix}(q)/\sum_{x\in\{0,1\}^n}q(x)$, which is ${\rm Alg}(q)$, is a probability distribution and
$$||q-{\rm Fix}(q)||_1 = 2\sum_{x\in\{0,1\}^n,q(x)<0} |q(x)|.$$
As described above, we can assume that $\sum_{x\in\{0,1\}^n}q(x)=1$, and thus can deal with the 
situation where ${\rm Fix}(q)={\rm Alg}(q)$. Therefore, Bremner et al.'s analysis of ${\rm Fix}(q)$ directly works for ${\rm Alg}(q)$. 
Moreover, to prove Theorem~\ref{basic}, it suffices to show the following lemma, which corresponds to a 
special case of the property of ${\rm Alg}(q)$ shown by Bremner et al.~\cite{Bremner3}:
\begin{lemma}\label{function}
We assume that there exists a constant $0 < \delta <1$ such that 
$||\widetilde{p}_{\rm A}-q||_1 \leq \delta/3$. Then, $||\widetilde{p}_{\rm A}-{\rm Alg}(q)||_1\leq \delta$.
\end{lemma}
\begin{proof} We upper-bound the value $||\widetilde{p}_{\rm A}-{\rm Alg}(q)||_1$ as follows:
\begin{align*}
||\widetilde{p}_{\rm A}-{\rm Alg}(q)||_1 & \leq ||\widetilde{p}_{\rm A}-q||_1 + ||q-{\rm Alg}(q)||_1 
\leq \frac{\delta}{3} + 2\sum_{x\in\{0,1\}^n,q(x)<0} |q(x)|\\
& \leq
\frac{\delta}{3}+2\sum_{x\in\{0,1\}^n,q(x)<0} |\widetilde{p}_{\rm A}(x)-q(x)|\leq \frac{\delta}{3} 
+ 2||\widetilde{p}_{\rm A}-q||_1 \leq \delta,
\end{align*}
which is the desired upper bound.
\end{proof}

Combining the above lemmas immediately implies Theorem~\ref{basic}:
\begin{proof}[Proof of Theorem~\ref{basic}]
By Lemma~\ref{coefficients}, there exists a polynomial-time randomized algorithm which outputs the Fourier 
coefficients of a function $q$ over $\{0,1\}^n$ such that 
$${\rm Pr}\left[||\widetilde{p}_{\rm A}-q||_1 \leq \frac{\delta}{3}\right] \geq 1-\frac{1}{\exp(n)}.$$
We define $\widetilde{q}_{\rm A}={\rm Alg}(q)$. By Lemma~\ref{function}, when $||\widetilde{p}_{\rm A}-q||_1 \leq \delta/3$, 
it holds that
$$||\widetilde{p}_{\rm A}-\widetilde{q}_{\rm A}||_1 = ||\widetilde{p}_{\rm A}-{\rm Alg}(q)||_1 \leq \delta.$$ 
By the sampling algorithm described above, $\widetilde{q}_{\rm A}$ is classically samplable in polynomial time. 
\end{proof}

\subsection{Applications of Theorem~\ref{basic}}\label{applications}

We first deal with an input-noise model where $D_{\varepsilon}$ is applied to each qubit (initialized to $|0\rangle$) only 
at the start of computation, and consider IQP circuits under this input-noise model as depicted in Fig.~\ref{figure2}(a). 
We show that, when noise model~{\bf A} is replaced with this input-noise model, Theorem~\ref{basic} holds for IQP circuits. 
Let $C=H^{\otimes n}DH^{\otimes n}$ be an arbitrary IQP circuit on $n$ qubits, where $D$ is a polynomial-size quantum circuit 
consisting of $Z$, $CZ$, and $CCZ$ gates. Let $\widetilde{p}_{\rm in}$ be the resulting probability distribution over $\{0,1\}^n$. 
The input state $|0^n\rangle$ affected by noise is represented as
$$\left[\left(1-\frac{\varepsilon}{2}\right)|0\rangle\langle0|+\frac{\varepsilon}{2}|1\rangle\langle1|\right]^{\otimes n}
=\sum_{y\in\{0,1\}^n}\left(1-\frac{\varepsilon}{2}\right)^{n-|y|}\left(\frac{\varepsilon}{2}\right)^{|y|}X^y|0^n\rangle\langle0^n|X^y.
$$
A direct calculation shows that $\widetilde{p}_{\rm in} = \widetilde{p}_{\rm A}$, where 
$\widetilde{p}_{\rm A}$ is the output probability distribution of $C$ under noise model~{\bf A} (with rate $\varepsilon$). 
This is because $H^{\otimes n}DH^{\otimes n}X^y = X^yH^{\otimes n}DH^{\otimes n}$ for any $y\in\{0,1\}^n$. Thus, when 
the output probability distribution of $C$ is anti-concentrated, $\widetilde{p}_{\rm A}$ is classically simulatable 
by Theorem~\ref{basic}, and so is $\widetilde{p}_{\rm in}$. Similarly, when noise model~{\bf B} is replaced 
with the input-noise model, Theorem~\ref{general} holds for IQP circuits. To show this, it suffices to 
represent the input state $|0^n\rangle$ affected by noise as
$$\bigotimes_{j=1}^n\left[\left(1-\frac{\varepsilon_j}{2}\right)|0\rangle\langle0|+\frac{\varepsilon_j}{2}|1\rangle\langle1|\right]
=\sum_{y\in\{0,1\}^n}\left[\prod_{j=1}^n\left(1-\frac{\varepsilon_j}{2}\right)^{1-y_j}\left(\frac{\varepsilon_j}{2}\right)^{y_j}\right]
X^y|0^n\rangle\langle0^n|X^y
$$
and to apply the relation $H^{\otimes n}DH^{\otimes n}X^y = X^yH^{\otimes n}DH^{\otimes n}$ as above.

\begin{figure}
\centering
\includegraphics[scale=.37]{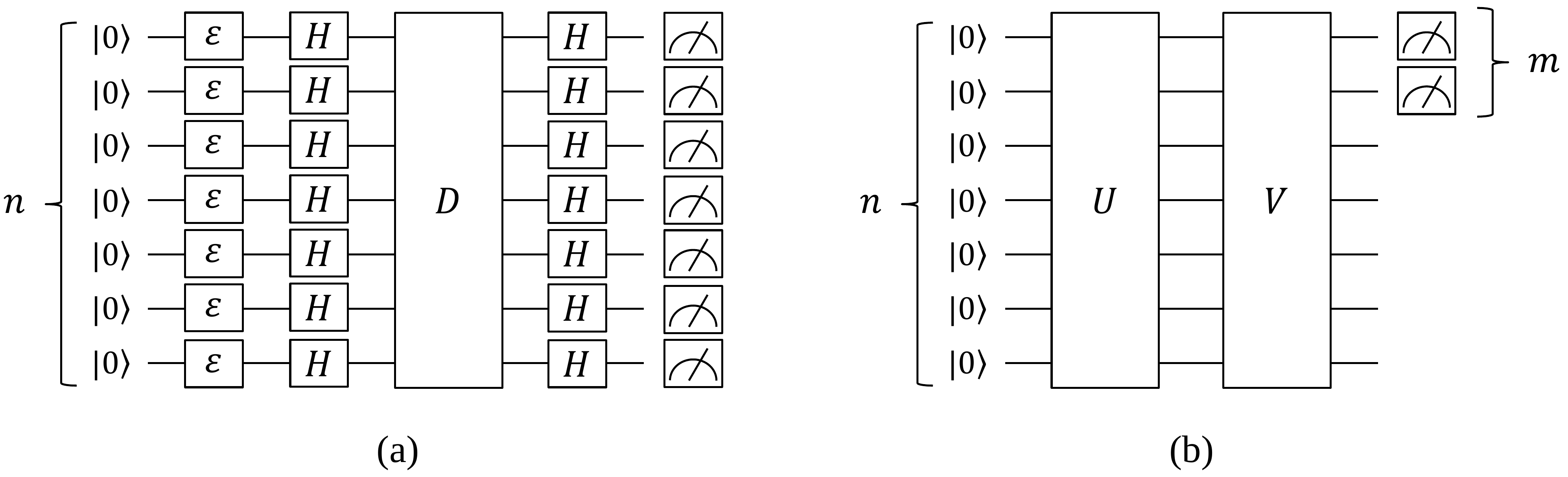}
\caption{(a): IQP circuit $C=H^{\otimes n}DH^{\otimes n}$ under the input-noise model, where $D$ is a 
polynomial-size quantum circuit consisting of $Z$, $CZ$, and $CCZ$ gates. The depolarizing 
channel $D_{\varepsilon}$, which is represented as $\varepsilon$ in this figure, is applied to each qubit 
initialized to $|0\rangle$. (b): CT-ECS circuit $C=VU$ in the noise-free setting, where $m=O(\log n)$.}
\label{figure2}
\end{figure}

We then consider CT-ECS circuits in the noise-free setting as depicted in Fig.~\ref{figure2}(b). 
Using a proof method similar to the one of Theorem~\ref{basic}, we show that, when only $O(\log n)$ qubits are measured, 
any quantum circuit in a class of CT-ECS circuits on $n$ qubits is classically simulatable in the noise-free setting. 
Let $C=VU$ be an arbitrary  CT-ECS circuit on $n$ qubits such that only $m=O(\log n)$ qubits are measured, and let $p$ be its 
output probability distribution over $\{0,1\}^m$. We can represent $p$ as
$$p(x)=\sum_{s \in\{0,1\}^m}\widehat{p}(s)(-1)^{s\cdot x}.$$ 
Since $m=O(\log n)$, roughly speaking, $p(x)$ can be computed as the sum of a polynomial number of the Fourier 
coefficients. A direct calculation similar to the proof of Lemma~\ref{fourier-rep} implies that
$$\widehat{p}(s)=\frac{1}{2^m}\langle 0^n|U^\dag V^\dag (Z^s \otimes I_{n-m}) VU|0^n\rangle$$
for any $s=s_1\cdots s_m \in\{0,1\}^m$, where $I_{n-m}$ is the identity on the qubits that are not measured. By an argument similar to the 
proof of Lemma~\ref{appro-fourier}, we want to approximate $\widehat{p}(s)$ for all $s\in\{0,1\}^m$. However, in the present case, 
$V^\dag (Z^s \otimes I_{n-m}) V$ is the product of at most $m$ ECS operations, which is not always ECS.

We assume that $V^\dag  (Z^s \otimes I_{n-m}) V$ is ECS for all $s\in\{0,1\}^m$. In this case, 
all the Fourier coefficients can be approximated classically in polynomial time with polynomial accuracy. Then, we define a 
function $q$ over $\{0,1\}^m$ as
$$q(x)=\sum_{s \in\{0,1\}^m}\widehat{p}'(s)(-1)^{s\cdot x},$$ 
where $\widehat{p}'(s)$ is the approximate value of $\widehat{p}(s)$ for all $s\in\{0,1\}^m$ and $\widehat{p}'(0^m)=1/2^m$. 
Since $||p-q||_1$ is upper-bounded by some inverse polynomial in $n$, the classical sampling 
algorithm described in Section~\ref{sampling} implies that $p$ is classically simulatable. More precisely, it 
can be approximated by a classically samplable probability distribution, which is ${\rm Alg}(q)$, with polynomial accuracy in $l_1$ norm.

For example, we consider a CT-ECS$_1$ circuit, which is a CT-ECS circuit such that the associated ECS operation 
$V^\dag Z_jV$ is ECS$_1$. As described at the end of Section~\ref{parallel1}, this means that, for any $x \in \{0,1\}^n$, 
$V^\dag Z_j V|x\rangle$ can be represented as a linear combination of at most $O(1)$ computational basis states. 
In this case, a simple calculation shows that $V^\dag  (Z^s \otimes I_{n-m}) V$ is ECS for all $s\in\{0,1\}^m$, since it is the product of at 
most $m$ ECS$_1$ (not ECS) operations. Thus, when only $O(\log n)$ qubits are measured, any CT-ECS$_1$ circuit 
on $n$ qubits is classically simulatable in the noise-free setting.

\section{CT-ECS Circuits under Noise Model~{\bf B}}\label{extension}

In this section, we prove Theorem~\ref{general}. Its precise statement is as follows:
\begin{theorem}
Let $C$ be an arbitrary CT-ECS circuit on $n$ qubits such that its output probability 
distribution $p$ over $\{0,1\}^n$ satisfies $\sum_{x\in \{0,1\}^n} p(x)^2 \leq \alpha/2^n$ 
for some known constant $\alpha \geq 1$. Let $0 < \delta < 1$ be an arbitrary constant. We assume that
\begin{itemize}
\item a depolarizing channel with (possibly unknown) constant rate $0 < \varepsilon_j < 1$ is applied to the $j$-th 
qubit after performing $C$ for any $1 \leq j \leq n$, which yields the probability 
distribution $\widetilde{p}_{\rm B}$ over $\{0,1\}^n$,
\item it is possible to choose a constant $\lambda_{\min}$ such that 
$$1 \leq \frac{\varepsilon_{\min}}{\lambda_{\min}} 
\leq 1+\frac{1}{\frac{10\sqrt{\alpha}}{\delta}\log\frac{10\sqrt{\alpha}}{\delta}},$$
where $\varepsilon_{\min} = \min \{\varepsilon_j |1 \leq j \leq n\}$, and

\item it is possible to choose a constant $\lambda_j$ such that 
$$0 \leq \varepsilon_j-\lambda_j  \leq 
\frac{\lambda_{\min}}{\frac{10\sqrt{\alpha}}{\delta}\log \frac{10\sqrt{\alpha}}{\delta}}$$
for any $1 \leq j \leq n$ with $\varepsilon_j \neq \varepsilon_{\min}$, where all numbers $j$ with 
$\varepsilon_j \neq \varepsilon_{\min}$ are known.
\end{itemize}
Then, there exists a polynomial-time randomized algorithm which outputs (a classical 
description of) a probability distribution $\widetilde{q}_{\rm B}$ over $\{0,1\}^n$ such that 
$${\rm Pr}\left[||\widetilde{p}_{\rm B}-\widetilde{q}_{\rm B}||_1 
\leq \left(1+\frac{1}{1-\lambda_{\min}}\right)\delta\right] \geq 1-\frac{1}{\exp(n)}$$ 
and $\widetilde{q}_{\rm B}$ is classically samplable in polynomial time.
\end{theorem}
\noindent
Although the approximate accuracy of $\widetilde{q}_{\rm B}$ depends on $\lambda_{\min}$, a typical situation might be 
when $\varepsilon_{\min}$ is not so large, such as $\varepsilon_{\min} \leq 1/2$, and thus $\lambda_{\min} \leq 1/2$. 
In this case, $||\widetilde{p}_{\rm B}-\widetilde{q}_{\rm B}||_1 \leq 3\delta$.

We represent the effect of noise under noise model~{\bf B} 
as the effect under noise model~{\bf A} with rate $\varepsilon_{\min}$ and 
the remaining effects. To do this, for any $1 \leq j \leq n$ and $0 \leq \delta_j \leq 1$, we consider the noise 
operator $T^j_{\delta_j}$~\cite{Odonnel-book} on real-valued functions over $\{0,1\}^n$ defined as
$$T^j_{\delta_j}f(x) = {\rm E}_{y_j \sim N_{\delta_j}(x_j)}[f(x_1,\cdots,x_{j-1},y_j,x_{j+1},\ldots,x_n)],$$
where $y_j \sim N_{\delta_j}(x_j)$ means that the random string $y_j$ is drawn as
$$
y_j =\left\{
\begin{array}{cc}
x_j & {\rm with \ probability} \ 1-\delta_j/2,\\
1-x_j & {\rm otherwise}
\end{array}\right.
$$
for any $x_j \in \{0,1\}$. 
Let $p$ be an arbitrary probability distribution over $\{0,1\}^n$. It is easy to verify that 
$T^1_{\delta_1}\cdots T^n_{\delta_n}p(x)$ is equal to the 
probability of obtaining $x$ by sampling an $n$-bit string according to $p$ and then flipping its $j$-th 
bit with probability $\delta_j/2$ for any $1 \leq j \leq n$. Thus, the representation of $\widetilde{p}_{\rm B}$, 
which is described in Section~\ref{fourier}, is obtained as $T^1_{\varepsilon_1}\cdots T^n_{\varepsilon_n}p(x)$. 
In particular, the representation of $\widetilde{p}_{\rm A}$ is obtained as $T^1_{\varepsilon}\cdots T^n_{\varepsilon}p(x)$. 
A key property of the operator 
$T^1_{\delta_1}\cdots T^n_{\delta_n}$ is that it is a contraction operator, i.e., 
$||T^1_{\delta_1}\cdots T^n_{\delta_n}f||_1 \leq ||f||_1$ for any real-valued function $f$ over $\{0,1\}^n$~\cite{Odonnel-book}. 
This can be simply shown by applying the triangle inequality for real numbers.

Combining Lemma~\ref{coefficients} with these basic facts on the noise operator, we show the following lemma on the 
representation of $\widetilde{p}_{\rm B}$ and its approximability:
\begin{lemma}\label{coin}
Let $C$ be an arbitrary CT-ECS circuit on $n$ qubits such that its output probability 
distribution $p$ satisfies $\sum_{x\in \{0,1\}^n} p(x)^2 \leq \alpha/2^n$ 
for some known constant $\alpha \geq 1$. Let $0 < \delta < 1$ be an arbitrary constant. We assume that
\begin{itemize}
\item a depolarizing channel with constant rate $0 < \varepsilon_j < 1$ is applied to the $j$-th 
qubit after performing $C$ for any $1 \leq j \leq n$, which yields the probability 
distribution $\widetilde{p}_{\rm B}$, and
\item it is possible to choose a constant $\lambda_{\min}$ such that 
$$1 \leq \frac{\varepsilon_{\min}}{\lambda_{\min}} \leq 1+ \frac{1}{\frac{10\sqrt{\alpha}}{\delta}\log\frac{10\sqrt{\alpha}}{\delta}},$$
where $\varepsilon_{\min} = \min \{\varepsilon_j |1 \leq j \leq n\}$.
\end{itemize}
Then, the following items hold:
\begin{itemize}
\item[\rm (i)] $\widetilde{p}_{\rm B}(x)=T^1_{\delta_1}\cdots T^n_{\delta_n}\widetilde{p}_{\rm A}(x)$, 
where
$$\delta_j=\frac{\varepsilon_j -\varepsilon_{\min}}{1 - \varepsilon_{\min}},\ \
\widetilde{p}_{\rm A}(x)=\sum_{s \in\{0,1\}^n}(1-\varepsilon_{\min})^{|s|}\widehat{p}(s)(-1)^{s\cdot x}.$$

\item[\rm (ii)] There exists a polynomial-time randomized algorithm which outputs the Fourier coefficients of 
a function $q$ over $\{0,1\}^n$ such that 
$${\rm Pr}\left[||\widetilde{p}_{\rm A}-q||_1 \leq \frac{\delta}{3}\right] \geq 1-\frac{1}{\exp(n)}.$$

\item[\rm (iii)] Let $0 \leq \delta_j' \leq 1$ be an arbitrary constant for any $1 \leq j \leq n$. We assume that there exists a 
constant $\beta \geq 0$ such that $|\delta_j - \delta_j'| \leq \beta$ for any $1 \leq j \leq n$. Then, 
$${\rm Pr}\left[||T^1_{\delta_1}\cdots T^n_{\delta_n}q -T^1_{\delta_1'}\cdots T^n_{\delta_n'}q||_1 \leq c\beta\sqrt{\alpha+1}\right]
\geq 1-\frac{1}{\exp(n)},$$
where $c = \left\lceil\frac{1}{\lambda_{\min}}\log \frac{10\sqrt{\alpha}}{\delta}\right\rceil$.
\end{itemize}
\end{lemma}
\begin{proof}
(i): As described above, the probability distribution $\widetilde{p}_{\rm B}$ is represented as
$$\widetilde{p}_{\rm B}(x)=T^1_{\varepsilon_1}\cdots T^n_{\varepsilon_n}p(x)
=\sum_{s \in\{0,1\}^n}\left[\prod_{j=1}^n  (1-\varepsilon_j)^{s_j}\right]\widehat{p}(s)(-1)^{s\cdot x}.$$
We can transform this representation of $\widetilde{p}_{\rm B}$ as follows:
\begin{align*}
\widetilde{p}_{\rm B}(x) & =\sum_{s \in\{0,1\}^n}\left[\prod_{j =1}^n \frac{(1-\varepsilon_j)^{s_j}}{(1-\varepsilon_{\min})^{s_j}}\right]
(1-\varepsilon_{\min})^{|s|}\widehat{p}(s)(-1)^{s\cdot x}\\
& = \sum_{s \in\{0,1\}^n}\left[ \prod_{j =1}^n
\left(1 - \frac{\varepsilon_j -\varepsilon_{\min}}{1 - \varepsilon_{\min}}\right)^{s_j}\right]
(1-\varepsilon_{\min})^{|s|}\widehat{p}(s)(-1)^{s\cdot x}\\
& = T^1_{\delta_1}\cdots T^n_{\delta_n}\sum_{s \in\{0,1\}^n}
(1-\varepsilon_{\min})^{|s|}\widehat{p}(s)(-1)^{s\cdot x} =T^1_{\delta_1}\cdots T^n_{\delta_n}\widetilde{p}_{\rm A}(x),
\end{align*}
which is the desired representation.\\
(ii): The probability distribution $\widetilde{p}_{\rm A}$ can be regarded as the output probability distribution of $C$ 
under noise model~{\bf A} with rate $\varepsilon_{\min}$. Since we have an approximate value $\lambda_{\min}$ 
of $\varepsilon_{\min}$, we can apply the proof of Lemma~\ref{coefficients}. More concretely, we define
$$c = \left\lceil\frac{1}{\lambda_{\min}}\log \frac{10\sqrt{\alpha}}{\delta}\right\rceil.$$
Since $ 0 < \lambda_{\min} < 1$ and $10\sqrt{\alpha}/\delta > 10$, $c >3$. We can 
obtain $\widehat{p}'(s)$ for all $s \in \{0,1\}^n\setminus \{0^n\}$ with $|s| \leq c$ such that
$${\rm Pr}\left[|\widehat{p}(s) - \widehat{p}'(s)| \leq \frac{1}{2^n f(n)}\right] \geq 1-\frac{1}{\exp(n)},$$
where $f(n)=10(n^c+1)/\delta$. We define $\widehat{p}'(0^n)=1/2^n$, which is equal to $\widehat{p}(0^n)$, and 
$\widehat{p}'(s)=0$ for all $s \in \{0,1\}^n$ with $|s| > c$. It holds that
$${\rm Pr}\left[\forall s \in \{0,1\}^n \ {\rm with} \ |s| \leq c,\ |\widehat{p}(s) - \widehat{p}'(s)| \leq \frac{1}{2^n f(n)}\right] 
\geq 1-\frac{1}{\exp(n)}.$$
Moreover, we define a function $q$ over $\{0,1\}^n$ as
$$q(x) = \sum_{s\in\{0,1\}^n, |s| \leq c}(1-\lambda_{\min})^{|s|}\widehat{p}'(s)(-1)^{s\cdot x}$$ 
and it holds that $||\widetilde{p}_{\rm A}-q||_1 \leq \delta/3$ 
under the assumption that, for all $s \in \{0,1\}^n$ with $|s| \leq c$, $|\widehat{p}(s) - \widehat{p}'(s)| \leq \frac{1}{2^n f(n)}$.\\
(iii): The function $T^1_{\delta_1}\cdots T^n_{\delta_n}q$ over $\{0,1\}^n$ can be represented as
$$T^1_{\delta_1}\cdots T^n_{\delta_n}q(x)
=\sum_{s:|s| \leq c}\left[\prod_{j=1}^n  (1-\delta_j)^{s_j}\right]\widehat{q}(s)(-1)^{s\cdot x},$$
where $\widehat{q}(s)=(1-\lambda_{\min})^{|s|}\widehat{p}'(s)$. Since 
$T^1_{\delta_1'}\cdots T^n_{\delta_n'}q$ can be represented similarly,
$$T^1_{\delta_1}\cdots T^n_{\delta_n}q(x) -T^1_{\delta_1'}\cdots T^n_{\delta_n'}q(x)
=\sum_{s:|s| \leq c}\left[\prod_{j=1}^n  (1-\delta_j)^{s_j}-\prod_{j=1}^n  (1-\delta_j')^{s_j}\right]\widehat{q}(s)(-1)^{s\cdot x}.$$
This representation with the relation described in Section~\ref{fourier} implies that
\begin{equation}\label{delta}
||T^1_{\delta_1}\cdots T^n_{\delta_n}q -T^1_{\delta_1'}\cdots T^n_{\delta_n'}q||_1^2 
\leq 2^{2n}\sum_{s:|s| \leq c}\left[\prod_{j=1}^n  (1-\delta_j)^{s_j}-\prod_{j=1}^n  (1-\delta_j')^{s_j}\right]^2\widehat{q}(s)^2.
\end{equation}
A simple calculation shows that, for any $s\in \{0,1\}^n$ with $s_{j_1}=\cdots =s_{j_{|s|}}=1$,
$$\left|\prod_{j=1}^n  (1-\delta_j)^{s_j}-\prod_{j=1}^n  (1-\delta_j')^{s_j}\right|
=\left|\prod_{k=1}^{|s|}  (1-\delta_{j_k})-\prod_{k=1}^{|s|}  (1-\delta_{j_k}')\right| \leq \sum_{k=1}^{|s|}|\delta_{j_k}-\delta_{j_k}'|.$$
Combining this with the inequality (\ref{delta}) and the assumption that $|\delta_j - \delta_j'| \leq \beta$ for any $1 \leq j \leq n$,
it holds that
$$||T^1_{\delta_1}\cdots T^n_{\delta_n}q -T^1_{\delta_1'}\cdots T^n_{\delta_n'}q||_1^2 
\leq 2^{2n}\sum_{s:|s| \leq c}(|s|\beta)^2\widehat{q}(s)^2 \leq c^2\beta^22^{2n}\sum_{s:|s| \leq c}\widehat{q}(s)^2.$$

The remaining problem is to show that $2^{2n}\sum_{s:|s| \leq c}\widehat{q}(s)^2 \leq \alpha+1$ under the assumption that, 
for all $s \in \{0,1\}^n$ with $|s| \leq c$, $|\widehat{p}(s) - \widehat{p}'(s)| \leq \frac{1}{2^n f(n)}$. 
This can be done by using the bounds $|\widehat{p}(s)| \leq 1/2^n$, $|\{s \in \{0,1\}^n ||s|\leq c\}|\leq n^c +1$, and 
$\sum_{x\in \{0,1\}^n} p(x)^2 \leq \alpha/2^n$ as follows:
\begin{align*}
& 2^{2n}\sum_{s:|s| \leq c}\widehat{q}(s)^2 
\leq  2^{2n}\sum_{s:|s| \leq c}\widehat{p}'(s)^2 \leq 2^{2n}\sum_{s:|s| \leq c}(|\widehat{p}(s)| + |\widehat{p}(s) - \widehat{p}'(s)|)^2 \\
& \leq 2^{2n}\sum_{s:|s| \leq c}\left(|\widehat{p}(s)| + \frac{1}{2^nf(n)}\right)^2 
= 2^{2n}\sum_{s:|s| \leq c}\left(\widehat{p}(s)^2 + \frac{2|\widehat{p}(s)|}{2^nf(n)}+\frac{1}{2^{2n}f(n)^2}\right) \\
& \leq  2^{n}\sum_{x\in \{0,1\}^n}p(x)^2 + \frac{2(n^c+1)}{f(n)}+\frac{n^c+1}{f(n)^2}
\leq \alpha +\frac{\delta}{5} +\frac{\delta^2}{100(n^c+1)} \leq \alpha+1,
\end{align*}
which is the desired bound.
\end{proof}

Lemma~\ref{coin} implies Theorem~\ref{general} as follows:
\begin{proof}[Proof of Theorem~\ref{general}]
Using the function $q$ obtained by Lemma~\ref{coin}, we define $\widetilde{q}_{\rm A}={\rm Alg}(q)$ 
and $\widetilde{q}_{\rm B}=T^1_{\delta_1'}\cdots T^n_{\delta_n'}\widetilde{q}_{\rm A}$, where
$$\delta_j'=\frac{\lambda_j -\lambda_{\min}}{1 - \lambda_{\min}}$$
for any $1 \leq j \leq n$ with $\varepsilon_j \neq \varepsilon_{\min}$ and $\delta_j'=0$ for any other $j$. 
In the following, we assume that, for all $s \in \{0,1\}^n$ with $|s| \leq c$, $|\widehat{p}(s) - \widehat{p}'(s)| \leq \frac{1}{2^n f(n)}$. 
In this case, $||\widetilde{p}_{\rm A}-q||_1 \leq \delta/3$. Thus, $||\widetilde{p}_{\rm A} - q||_1 + ||q - \widetilde{q}_{\rm A}||_1 \leq \delta$ 
by the proof of Lemma~\ref{function}. This implies that
\begin{align}\label{final}
||\widetilde{p}_{\rm B} - \widetilde{q}_{\rm B}||_1 
& = ||T^1_{\delta_1}\cdots T^n_{\delta_n}\widetilde{p}_{\rm A} - T^1_{\delta_1'}\cdots T^n_{\delta_n'}\widetilde{q}_{\rm A}||_1 \nonumber \\
& \leq ||T^1_{\delta_1}\cdots T^n_{\delta_n}\widetilde{p}_{\rm A} - T^1_{\delta_1}\cdots T^n_{\delta_n}q||_1
+ ||T^1_{\delta_1}\cdots T^n_{\delta_n}q - T^1_{\delta_1'}\cdots T^n_{\delta_n'}q||_1 \nonumber \\
& \hspace{.4cm} +||T^1_{\delta_1'}\cdots T^n_{\delta_n'}q - T^1_{\delta_1'}\cdots T^n_{\delta_n'}\widetilde{q}_{\rm A}||_1 \nonumber \\
& \leq ||\widetilde{p}_{\rm A} - q||_1 + ||T^1_{\delta_1}\cdots T^n_{\delta_n}q - T^1_{\delta_1'}\cdots T^n_{\delta_n'}q||_1
+||q - \widetilde{q}_{\rm A}||_1\nonumber \\
& \leq \delta + ||T^1_{\delta_1}\cdots T^n_{\delta_n}q - T^1_{\delta_1'}\cdots T^n_{\delta_n'}q||_1,
\end{align}
where the second inequality is due to the fact that the noise operator is a contraction operator as described at the beginning 
of this section.

It is obvious that $|\delta_j-\delta_j'|=0$ for any $1 \leq j \leq n$ with $\varepsilon_j = \varepsilon_{\min}$. We show that
$$|\delta_j-\delta_j'|\leq \frac{\delta}{1 - \lambda_{\min}}\cdot 
\frac{\lambda_{\min}}{10\sqrt{\alpha}\log \frac{10\sqrt{\alpha}}{\delta}}$$
for any $1 \leq j \leq n$ with $\varepsilon_j \neq \varepsilon_{\min}$. 
In fact, $\delta_j-\delta_j'$ is upper-bounded by using the bounds $\lambda_{\min} \leq \varepsilon_{\min} <1$ and 
$\varepsilon_j - \lambda_j \leq \frac{\lambda_{\min}}{\frac{10\sqrt{\alpha}}{\delta}\log \frac{10\sqrt{\alpha}}{\delta}}$ as follows:
\begin{align*}
\delta_j-\delta_j' 
& = \left(1 - \frac{1-\varepsilon_j}{1-\varepsilon_{\min}} \right) 
- \left(1 - \frac{1-\lambda_j}{1-\lambda_{\min}} \right) 
= \frac{1-\lambda_j}{1-\lambda_{\min}} - \frac{1-\varepsilon_j}{1-\varepsilon_{\min}}\\
& \leq \frac{1-\lambda_j}{1-\lambda_{\min}} - \frac{1-\varepsilon_j}{1-\lambda_{\min}}
= \frac{\varepsilon_j-\lambda_j}{1-\lambda_{\min}}
\leq  \frac{\delta}{1 - \lambda_{\min}}\cdot \frac{\lambda_{\min}}{10\sqrt{\alpha} \log \frac{10\sqrt{\alpha}}{\delta}}.
\end{align*}
On the other hand, $\delta_j-\delta_j'$ is lower-bounded by using the bounds $\lambda_{\min} \leq \varepsilon_{\min}
\leq \lambda_{\min}+\frac{\lambda_{\min}}{{\frac{10\sqrt{\alpha}}{\delta}\log \frac{10\sqrt{\alpha}}{\delta}}}$ and 
$\lambda_j \leq \varepsilon_j$ as follows:
\begin{align*}
\delta_j-\delta_j' 
& =  \frac{\varepsilon_j - \varepsilon_{\min}}{1-\varepsilon_{\min}} -\frac{\lambda_j-\lambda_{\min}}{1-\lambda_{\min}} 
 \geq  \frac{\varepsilon_j - \varepsilon_{\min}}{1-\lambda_{\min}} - \frac{\lambda_j -\lambda_{\min}}{1 - \lambda_{\min}}
 =\frac{\varepsilon_j - \lambda_j}{1-\lambda_{\min}} - \frac{\varepsilon_{\min} -\lambda_{\min}}{1 - \lambda_{\min}} \\
& \geq - \frac{\varepsilon_{\min} -\lambda_{\min}}{1 - \lambda_{\min}}
 \geq -\frac{\delta}{1 - \lambda_{\min}}\cdot \frac{\lambda_{\min}}{10\sqrt{\alpha} \log \frac{10\sqrt{\alpha}}{\delta}}.
\end{align*}

By the definition of $c$, it holds that $2 < c-1 \leq \frac{1}{\lambda_{\min}}\log \frac{10\sqrt{\alpha}}{\delta}$. Thus, 
the above upper bound of $|\delta_j-\delta_j'|$ implies that
$$|\delta_j-\delta_j'|\leq \frac{\delta}{1 - \lambda_{\min}}\cdot 
\frac{\lambda_{\min}}{10\sqrt{\alpha}\log \frac{10\sqrt{\alpha}}{\delta}}
\leq \frac{\delta}{1 - \lambda_{\min}}\cdot \frac{1}{10(c-1)\sqrt{\alpha}}$$
for any $1 \leq j \leq n$. By Lemma~\ref{coin} with $\beta =\frac{\delta}{1 - \lambda_{\min}}\cdot \frac{1}{10(c-1)\sqrt{\alpha}}$,
$$||T^1_{\delta_1}\cdots T^n_{\delta_n}q - T^1_{\delta_1'}\cdots T^n_{\delta_n'}q||_1
\leq \frac{\delta}{1 - \lambda_{\min}}\cdot \frac{c\sqrt{\alpha+1}}{10(c-1)\sqrt{\alpha}}
\leq \frac{\delta}{1 - \lambda_{\min}},$$
since $c > 3$ and $\alpha \geq 1$. This bound with the inequality (\ref{final}) implies that
$$||\widetilde{p}_{\rm B} - \widetilde{q}_{\rm B}||_1 
\leq \delta + \frac{\delta}{1 - \lambda_{\min}} = \left(1+ \frac{1}{1 - \lambda_{\min}}\right)\delta.$$

The above analysis means that, to simulate $\widetilde{p}_{\rm B}$, it suffices to sample $\widetilde{q}_{\rm B}$, i.e., 
to sample $\widetilde{q}_{\rm A}$ and flip its $j$-th bit using a biased coin with probability
$$\frac{\delta_j'}{2}=\frac{\lambda_j -\lambda_{\min}}{2(1 - \lambda_{\min})}
=\frac{1}{2}\left(1-\frac{1-\lambda_j}{1 - \lambda_{\min}}\right)$$
of heads for any $1 \leq j \leq n$ with $\varepsilon_j \neq \varepsilon_{\min}$. By the sampling algorithm described in 
Section~\ref{sampling}, $\widetilde{q}_{\rm A}$ is classically samplable in polynomial time. Thus, $\widetilde{q}_{\rm B}$ is also 
classically samplable in polynomial time. We note that, although it may not be possible to flip such a biased coin with perfect 
accuracy, it suffices to flip a biased coin whose probability of heads is exponentially close to that of the above coin by computing the 
value $\frac{1-\lambda_j}{1 - \lambda_{\min}}$ with exponential accuracy.
\end{proof}

\section{Conclusions and Future Work}

We considered the effect of noise on the classical simulatability of CT-ECS circuits, such as IQP, Clifford Magic, 
conjugated Clifford, and constant-depth quantum circuits. We showed that, under noise model~{\bf A}, 
if an approximate value of the noise rate is known, any CT-ECS circuit with an anti-concentrated output probability distribution 
is classically simulatable. This indicates that the presence of small noise drastically affects the classical simulatability of CT-ECS 
circuits. We also considered noise model~{\bf B} where the noise rate can vary 
with each qubit, and provided a similar sufficient condition for classically simulating CT-ECS circuits with 
anti-concentrated output probability distributions.

Interesting challenges would be to investigate the effect of noise on the classical simulatability of other 
quantum circuits that are not classically simulatable in the noise-free setting (under plausible assumptions), 
such as quantum circuits for Shor's factoring 
and discrete logarithm algorithms~\cite{Shor}. At present, it seems difficult to apply our analysis directly to such quantum circuits 
based on the quantum Fourier transform. However, there exists a classical algorithm for simulating such circuits with sparse 
output probability distributions in the noise-free setting~\cite{Schwarz} and this algorithm might be 
useful in a noise setting. It would also be interesting to investigate fault-tolerant schemes to 
protect CT-ECS circuits from noise. As described in~\cite{Bremner3}, a simple repetition 
code can be used to protect IQP circuits under noise model~{\bf A} with a known rate, and the 
circuits produced by using the error-correcting code are also IQP circuits. Similarly, the simple repetition 
code can be used for CT-ECS circuits under noise model~{\bf A} with a known rate. However, 
the produced circuits are in general completely different from the original ones. This would decrease the possibility of 
implementing such circuits, and thus it would be interesting to study how to avoid this problem.



\bibliography{mybib}

\end{document}